\documentclass[runningheads,a4paper]{llncs}
\usepackage{amssymb}
\usepackage{amsmath}
\usepackage{stmaryrd}
\usepackage{graphicx}
\usepackage{hyperref}
\newcommand{\dom}{\mathrm{dom}}
\newcommand{\codom}{\mathrm{codom}}
\newcommand{\arc}[1]{\mathop{\longrightarrow}\limits^{#1}}
\newcommand{\bra}[1]{\langle #1 |}
\newcommand{\ket}[1]{| #1\rangle}
\newcommand{\braket}[2]{\langle #1 | #2\rangle}
\newcommand{\trace}[1]{\mathrm{Tr}(#1)}
\newcommand{\prob}[2]{\mathrm{Pr}(#1~|~#2)}
\newcommand{\effect}[2]{\mathrm{Eff}[#1~|~#2]}
\newcommand{\matid}{\mathbf{1}}
\begin{document}
\title{Towards the Notion of an Abstract Quantum Automaton}
\author{Mizal Alobaidi\inst{1} \and
Andriy Batyiv\inst{2} \and 
Grygoriy Zholtkevych\inst{2}}
\institute{Tikrit University, Faculty of Computer Sciences and Mathematics,\\
P.O.~Box--42, Tikrit, Iraq,\\
\email{mizalobaidi@yahoo.com} \and
V.N. Karazin Kharkiv National University,
School of Mathematics\\
and Mechanics, 4, Svobody Sqr., Kharkiv, 61077, Ukraine,\\
\email{generatorglukoff@gmail.com, zholtkevych@univer.kharkov.ua}}
%\authorrunning{M.~Alobaidi, A.~Batyiv, G.~Zholtkevych}
%
\maketitle
\begin{abstract}
The main goal of this paper is to give a rigorous mathematical description of systems for processing quantum information.
To do it authors consider abstract state machines as models of classical computational systems.
This class of machines is refined by introducing constrains on a state structure, namely, it is assumed that state of computational process has two components: a control unit state and a memory state. 
Then authors modify the class of models by substituting the deterministic evolutionary mechanism for a stochastic evolutionary
mechanism.
This approach can be generalized to the quantum case: one can replace  transformations of a classical memory with quantum
operations on a quantum memory.
Hence the authors come to the need to construct a mathematical model of an operation on the quantum memory.
It leads them to the notion of an abstract quantum automaton.
Further the authors demonstrate that a quantum teleportation process is described as evolutionary process for some abstract
quantum automaton.
\end{abstract}
% --------------------------------------------------------------------------------------------------------------------------------
\section*{Introduction}
The idea to build a device capable "to compute all that can be computed" had emerged a long time.
One can remember Blaise Pascal's Arithmetic Machine, Gottfried Wilhelm Leibniz's Stepped Reckoner, Charles Babbage's Difference
Engine and his Analytical Engine \cite{bib::white}.
But only in the thirties of last century, Alonzo Church \cite{bib::church}, Alan Mathison Turing \cite{bib::turing_machine}, and
Emil Leon Post \cite{bib::post} built mathematical models of the computational processes.
Although these models have different shapes each of them describes inherently the same class of processes.
The equivalence of the Turing's model and the Church's model, for example, was proved by A.M. Turing in 1937
\cite{bib::turing_thesis}.
In the late forties, hardware implementations of a universal computational system were developed and began to be used.
They are known now as computers.

The practice of using computers to solve real problems showed that besides answering the question "Can the problem be solved using computer?", an answer to the question "Do we have enough computational resources to solve the problem?" is important too.
Searches for methods to evaluate computational resources for computer-assisted problem solving led to the special scientific area
which is called theory of computational complexity (the brief historical overview one can see in \cite{bib::fortnow}).
Unfortunately, most important computational problems are complex ones.
In compliance with the generally accepted propositions of theoretical computing science the application field of classical
computers, i.e. hardware implementations of the universal Turing machine concept, is physically challenged by problems which have polynomial computational complexity.

However, modern science, technique, and technology are in need of me\-thods to solve problems whose complexity is higher than polynomial.
This situation stimulates research of non-classical approaches to computing, and quantum computing is one of these.

The idea to use quantum systems as computing devices appeared in the early eighties of the twentieth century.
The idea's authors considered it as a way to overcome computational complexity.
In the context Yuri Ivanovitch Manin's monograph
\footnote{
	The monograph's introduction was translated into English \cite{bib::manin_essays}
}
\cite{bib::manin} and Richard Phillips Feynman's paper \cite{bib::feynman} should be noted.
Considering the possibility of using quantum machines for solving complex problems of simulation Yu.I. Manin wrote (cited by \cite{bib::manin_essays}): "\dots we need a mathematical theory of quantum automata. Such a theory would provide us with
mathematical models of deterministic processes with quite unusual properties. One reason for this is that the quantum state space
has far greater capacity then the classical one: for a classical system with $N$ states, its quantum version allowing
superposition (entanglement) accommodates $e^N$ states".
In \cite{bib::manin_essays}, Yu.I. Manin also sets requirements to the mathematical theory of quantum automata: "The first
difficulty we must overcome is the choice of the correct balance between the mathematical and the physical principles. The quantum automaton has to be an abstract one: its mathematical model must appeal only to the general principles of quantum physics, without prescribing a physical implementation. Then the model of evolution is the unitary rotation in a finite dimensional Hilbert space,
and the decomposition of the system into its virtual parts corresponds to the tensor product decomposition of the state space
("quantum entanglement"). Somewhere in this picture we must accommodate interaction, which is described by density matrices and probabilities".

R.P. Feynman had a similar opinion \cite{bib::feynman,bib::feynman_computer}.

This paper is an attempt to construct a mathematical model of quantum automata  that fulfils requirements formulated by
Yu.I. Manin.
% --------------------------------------------------------------------------------------------------------------------------------
\section{Classical Computational Model}
In this section a mathematical model of a classical computational system is considered.
The approach was proposed by A.N.~Kolmogorov and V.A.~Uspensky \cite{bib::kolmogorov}.
Theory of abstract state machine (ASM) is the further development of the approach \cite{bib::gurevich_1,bib::gurevich_2}.
% --------------------------------------------------------------------------------------------------------------------------------
\subsection{Preliminary definitions}
\begin{definition}\label{dfn::algorithm}
Let $\mathfrak{A}$ denotes an {\bfseries algorithm}. It is determined by
\begin{itemize}
	\item[$-$]
		a set $\mathcal{C}(\mathfrak{A})$ of states{\rm;}
	\item[$-$]
		a subset $\mathcal{I}(\mathfrak{A})$ of $\mathcal{C}(\mathfrak{A})$ which elements are called {\bfseries initial} states;
	\item[$-$]
		a subset $\mathcal{T}(\mathfrak{A})$ of $\mathcal{C}(\mathfrak{A})$ which elements are called {\bfseries terminal} states;
	\item[$-$]
		a map $\tau_\mathfrak{A}\colon\mathcal{C}(\mathfrak{A})\rightarrow\mathcal{C}(\mathfrak{A})$ which defines one step of
		the computational process;
\end{itemize}
and the next condition
\begin{equation}
\mathcal{I}(\mathfrak{A})\bigcap\mathcal{T}(\mathfrak{A})=\emptyset.
\end{equation}
\end{definition}
Note that elements of the set $\mathcal{C}(\mathfrak{A})$ correspond to complete state descriptions of the computational process
which is defined by the algorithm $\mathfrak{A}$. 
\begin{definition}\label{dfn::run}
Let $\mathfrak{A}$ be an algorithm then a partial map $C\colon\mathbb{N}\dashrightarrow\mathcal{C}(\mathfrak{A})$
\footnote{
For two sets $A$ and $B$ by $f\colon A\dashrightarrow B$ a partial map from $A$ into $B$ is denoted.
For $a\in A$ by $f(a)\neq\emptyset$ the clause "$f(a)$ is defined" is denoted.
}
is called a {\bfseries run} of the algorithm if it satisfies the following conditions
\begin{itemize}
	\item[{\rm(i)}]
		$C(0)\in\mathcal{I}(\mathfrak{A});$
	\item[{\rm(ii)}]
		if $C(t)\neq\emptyset$ for some $t\in \mathbb{N}$ then $C(t')\neq\emptyset$ for all $t'\in \mathbb{N}$ such that $t'<t;$
	\item[{\rm(iii)}]
		if $C(t+1)\neq\emptyset$ for some $t\in \mathbb{N}$ then $C(t+1)=\tau_\mathfrak{A}(C(t))${\rm;} 
	\item[{\rm(iv)}]
		if $\emptyset\neq C(t)\in\mathcal{T}(\mathfrak{A})$ then $C(t+1)=\emptyset$.
\end{itemize}
\end{definition}
From this definition it follows immediately that the domain of an arbitrary run is the set $\mathbb{N}$ or some set
$0~..~T=\{t\in\mathbb{N}\mid t\leq T\}$, where $T$ is a non-negative integer.

In the first case the algorithm {\bf diverges} on the initial state $C(0)$ (this is denoted by $\mathfrak{A}(C(0))\uparrow$).

In the second case the algorithm {\bf converges} on the initial state $C(0)$ to $C(T)$(this is denoted by
$\mathfrak{A}(C(0))\downarrow C(T)$).
% --------------------------------------------------------------------------------------------------------------------------------
\subsection{Abstract state machines with stochastic behaviour}
Let's refine the Definition \ref{dfn::algorithm} and Definition \ref{dfn::run} for such algorithms that have sets of states with some special structure.

Let's start refining with the following auxiliary definitions.
\begin{definition}
Let $\mathcal{N}$ and $\mathcal{A}$ be finite sets of nodes and arcs respectively, $\dom$ and $\codom$ be a maps that associate
with arcs their initial and terminal nodes respectively, then the tuple $(\mathcal{N},\mathcal{A},\dom,\codom)$ is called a
{\bfseries directed multigraph}.
\end{definition}
\begin{definition}
Let $\mathcal{G}=(\mathcal{N},\mathcal{A},\dom,\codom)$ be a directed multigraph, then an alternating sequence
$\alpha=n_1,a_1,\dots,a_k,n_{k+1}$ of nodes and arcs, beginning and ending with a node, is called a {\bfseries walk} if for all
$s=1,\dots,k$ the next condition holds: $\dom(a_s)=n_s$ and $\codom(a_s)=n_{s+1}$.

In this case we shall use the notation: $n_1\arc{a_1}\dots\arc{a_k}n_{k+1}$.
\end{definition}
\begin{definition}
Let $\alpha=n_1\arc{a_1}\dots\arc{a_k}n_{k+1}$ be a walk in the directed multigraph
$\mathcal{G}=(\mathcal{N},\mathcal{A},\dom,\codom)$ and $n$ be its node, then we shall say that
\begin{itemize}
	\item[{\rm(i)}]
		$n$ is the {\bfseries initial node} of $\alpha$ {\rm(}it is denoted by $\dom(\alpha)=n${\rm)} if $n=n_1${\rm;}
	\item[{\rm(ii)}]
		$n$ is the {\bfseries terminal node} of $\alpha$ {\rm(}it is denoted by $\codom(\alpha)=n${\rm)} if $n=n_{k+1}${\rm;}
	\item[{\rm(iii)}]
		$\alpha$ {\bfseries traverses} $n$ if for some $s\in\{1,\dots,k+1\}$ the equality $n=n_s$ holds.
\end{itemize}
\end{definition}
\begin{definition}
Let $(\mathcal{N},\mathcal{A},\dom,\codom,n_0,F)$ be a tuple such that the tuple\linebreak
$(\mathcal{N},\mathcal{A},\dom,\codom)$ is a directed multigraph, $n_0$ is a fixed node {\rm(}it is called the initial node{\rm)},
$F$ is a fixed subset of nodes{\rm(}its elements are called terminal nodes{\rm)}. The tuple is called a {\bfseries control graph}
if the next conditions hold:
\begin{itemize}
	\item[{\rm(i)}]
		$n_0\notin F${\rm;}
	\item[{\rm(ii)}]
		for each $n\in F$ there is no arc with initial node equals $n${\rm;}
	\item[{\rm(iii)}]
		for each node $n$ there is a walk such that its initial node equals $n_0$, its  terminal node belongs to $F$, and it traverses $n$.
\end{itemize} 
\end{definition}

Note that for the control graph $(\mathcal{N},\mathcal{A},\dom,\codom,n_0,F)$ and each $n\in\mathcal{N}\setminus F$ the set  $\mathrm{Out}(n)=\{a\in\mathcal{A}\mid\dom(a)=n\}$ is not empty.

Let's assume that for an arbitrary algorithm $\mathfrak{A}$ the set of states has the next structure
$\mathcal{C}(\mathfrak{A})=\mathcal{N}(\mathfrak{A})\times\mathcal{S}(\mathfrak{A})$ where $\mathcal{N}(\mathfrak{A})$ is the
nodes set of some control graph $\mathcal{G}(\mathfrak{A})$ and $\mathcal{S}(\mathfrak{A})$ is some set of memory snapshots.

In this case suppose that the set of initial states is the next set
$\mathcal{I}(\mathfrak{A})=\{(n_0,S)\mid S\in\mathcal{S}(\mathfrak{A})\}$ and the set of terminal states is the following set
$\mathcal{T}(\mathfrak{A})=\{(n,S)\mid n\in F~\&~S\in\mathcal{S}(\mathfrak{A})\}$.

This supposition leads us to the next representation of the map $\tau_\mathfrak{A}$:
\begin{multline}\label{eqn::tau1}
\tau_\mathfrak{A}(n,S)=(\sigma_\mathfrak{A}(n,S),\gamma_\mathfrak{A}(n,S))\mbox{, where }\\
\sigma_\mathfrak{A}\colon\mathcal{N}(\mathfrak{A})\times\mathcal{S}(\mathfrak{A})\rightarrow\mathcal{N}(\mathfrak{A})\mbox{ and }
\gamma_\mathfrak{A}\colon\mathcal{N}(\mathfrak{A})\times\mathcal{S}(\mathfrak{A})\rightarrow\mathcal{S}(\mathfrak{A}). 
\end{multline} 

Suppose now that the map $\sigma_\mathfrak{A}$ has property of locality.
It means that for each $n\in\mathcal{N}(\mathfrak{A})\setminus F$ there exists a map
$h_n\colon\mathcal{S}(\mathfrak{A})\rightarrow\mathrm{Out}(n)$ and for each $a\in\mathcal{A}(\mathfrak{A})$
\footnote{By $\mathcal{A}(\mathfrak{A})$ is denoted the set of arcs for the control graph $\mathcal{G}(\mathfrak{A})$.}
there exists a map $g_a\colon\mathcal{S}(\mathfrak{A})\rightarrow\mathcal{S}(\mathfrak{A})$
such that the following equalities are true:
\begin{eqnarray}
\sigma_\mathfrak{A}(n,S)=&\codom(h_n(S));\label{eqn::sigma}\\
\gamma_\mathfrak{A}(n,S)=&g_{h_n(S)}(S)\label{eqn::gamma}.
\end{eqnarray}

From (\ref{eqn::tau1}), (\ref{eqn::sigma}), and (\ref{eqn::gamma}) it follows
\begin{equation}\label{eqn::tau2}
\tau_\mathfrak{A}(n,S)=(\codom(h_n(S)),g_{h_n(S)}(S)).
\end{equation}

Therefore, we can consider the computational process which is determined by the algorithm $\mathfrak{A}$ as a sequence of steps.
Each step begins when the current state is described by some control graph node $n$ and a memory snapshot $S$.
Then the map $h_n$ chooses the arc $a$ outgoing from the node $n$ depending on the snapshot $S$.
Finally, using the selected arc and the memory snapshot the new control graph node and the new memory snapshot
are determined in compliance with (\ref{eqn::tau2}).

Let's modify the computational model by rejecting the assumption about determinacy for the choosing process.
Definition \ref{dfn::asmr} describes this modification formally.
\begin{definition}\label{dfn::asmr}
Let $\mathcal{G}=(\mathcal{N},\mathcal{A},\dom,\codom,n_0,F)$ be a control graph, $\mathcal{S}$ be some set of memory snapshots, $\mathcal{P}=\{\prob{\cdot}{S,n}\mid n\in\mathcal{N}\setminus F,S\in\mathcal{S}\}$ be a family of probability distributions on $\mathcal{A}$, and $\mathcal{T}=\{g_a\mid a\in\mathcal{A}\}$ be a family of maps from $\mathcal{S}$ into itself then the tuple $(\mathcal{G},\mathcal{S},\mathcal{P},\mathcal{T})$ is called an {\bfseries abstract state machine with stochastic behaviour} if
the following condition holds
\begin{multline}
\mbox{for all }n\in\mathcal{N},S\in\mathcal{S},\mbox{ and }a\in\mathcal{A}\\
\prob{a}{S,n}=0\mbox{ follows from }a\notin\mathrm{Out}(n). 
\end{multline}
\end{definition}

Dynamics of such machines is determined by the next definition.
\begin{definition}\label{dfn::rrun}
Let $(\mathcal{G},\mathcal{S},\mathcal{P},\mathcal{T})$ be an abstract state machine with stochastic behaviour,
$\mathcal{G}=(\mathcal{N},\mathcal{A},\dom,\codom,n_0,F)$ be its control graph then a partial map $C\colon\mathbb{N}\dashrightarrow\mathcal{N}\times\mathcal{A}$ is called a {\bfseries run} of the machine if it satisfies the
following conditions
\begin{itemize}
	\item[{\rm(i)}]
		$C(0)=(n_0,S)\mbox{, where }S\in\mathcal{S}${\rm;}
	\item[{\rm(ii)}]
		if $C(t)\neq\emptyset$ for some $t\in \mathbb{N}$ then $C(t')\neq\emptyset$ for all $t'\in \mathbb{N}$ such that $t'<t$
		{\rm;}
	\item[{\rm(iii)}]
		if $\emptyset\neq C(t+1)=(n',S')$ for some $t\in \mathbb{N}$ and $C(t)=(n,S)$ then there exists $a\in\mathrm{Out}(n)$ such
		that $\prob{a}{n,S}>0$, $n'=\codom(a)$, and $S'=g_a(S)${\rm;} 
	\item[{\rm(iv)}]
		if $\emptyset\neq C(t)=(n,S)$ for $n\in F$ and $S\in\mathcal{S}$ then $C(t+1)=\emptyset$.
\end{itemize}
\end{definition}

Below such machines will be generalised for the quantum case.
% --------------------------------------------------------------------------------------------------------------------------------
\section{Mathematical Model of Finite-Level Quantum Systems}
In the section the model of quantum systems with finite quantity of levels (finite-level quantum systems) is described.
It is based on the approaches set forth in the works \cite{bib::holevo,bib::nielsen}.
% --------------------------------------------------------------------------------------------------------------------------------
\subsection{Postulates of finite-level quantum systems}
The postulates of finite-level quantum systems fix basic notions which are used to construct mathematical models for the systems.

{\bfseries Postulate 1:} an $n$-dimensional Hilbert space $\mathcal{H}_n$ is associated to any quantum physical system with $n$
levels.
This space is known as the state space of the system.
The system is completely described by its pure state, which is a one-dimensional subspace of the state space.
This subspace is uniquely represented by the ortho-projector $\ket{\psi}\bra{\psi}$ on a vector $\ket{\psi}$ which generates the subspace.

In contrast to pure states mixed states are used to describe quantum systems whose state is not completely known.

Rather more detailed suppose we know that a quantum system is in one of a number of states
$\{\ket{\psi_k}\bra{\psi_k}~\colon~k=1,\dots,m\}$ with respective probabilities $\{p_k\colon k=1,\dots,m\}$. 
We shall call $\{p_k,\ket{\psi_k}\bra{\psi_k}~\colon~k=1,\dots,m\}$ an ensemble of pure states.
The density operator for the system is defined by the equation
\[ 
\rho=\sum\limits_{k=1}^m p_k\ket{\psi_k}\bra{\psi_k}.   
\]

We identify mixed states with density operators
\footnote{The set of all density operators on the space $\mathcal{H}_n$ is denoted by $\mathfrak{S}_n$}.
Evidently, that each density operator is a non-negative defined operator which trace is equal to unit.
It is known that the inverse statement is true: a non-negative defined operator, which trace is equal to unit, is a density
operator \cite{bib::holevo}.

Of course, a one-dimensional ortho-projector is a density operator.
The set of all density operators is convex and its subset of one-dimensional ortho-projectors is the subset of its extreme points \cite{bib::holevo}.
This allows to consider pure states as indecomposable states.

{\bfseries Postulate 2:} the state space of a composite physical system is the tensor product of the state spaces of the component physical systems.
Moreover, if we have systems indexed by $k=1,\dots,m$, and the state of the system with number $k$ is described by the density operator $\rho_k$, then the joint state of the total system before any interactions is $\rho_1\otimes\cdots\otimes\rho_m$.

{\bfseries Postulate 3:} the evolution of a closed quantum system is described by a unitary transformation. That is, the state $\ket{\psi}\bra{\psi}$ of the system at time $t_1$ is related to the state $\ket{\psi '}\bra{\psi '}$ of the system at time $t_2$
by a unitary operator $U$ which depends only on the times $t_1$ and $t_2$,
\[
\ket{\psi '}\bra{\psi '}=U\ket{\psi}\bra{\psi}U^\ast.
\] 

If we have an ensemble of pure states of the system which is described by the density operator $\rho$ at time $t_1$ then the density operator $\rho '$ of the system at time $t_2$ can be calculated by the formula
\[
\rho '=U\rho U^\ast.
\]

{\bfseries Postulate 4:} quantum measurements are described by an indexed finite family $\mathcal{K}=\{K(x)\colon x\in X\}$ of Kraus' operators, where $X$ is a finite set.
These are operators acting on the state space of the system being measured.
The index $x$ refers to the measurement outcome that may occur in the experiment.
If the state of the quantum system is described by the density operator $\rho$ immediately before the measurement then the
probability that result $x$ occurs is given by the following formula
\footnote{
	By $\trace{\cdot}$ the usual operator trace is denoted
}
\begin{equation}\label{frm::prob}
\prob{x}{\rho}=\trace{\rho K(x)^\ast K(x)}
\end{equation}
and the state of the system immediately after the measurement is described by the density operator
\begin{equation}\label{frm::effect}
\effect{\rho}{x}=\frac{K(x)\rho K(x)^\ast}{\trace{\rho K(x)^\ast K(x)}}
\end{equation}
Any Kraus' family $\mathcal{K}=\{K(x)\colon x\in X\}$ satisfies the {\bfseries completeness condition}
\begin{equation}\label{frm::comp}
\sum\limits_{x\in X}K(x)^\ast K(x)=\matid,
\end{equation}
which ensures correctness of the definitions given by formulas (\ref{frm::prob}) and (\ref{frm::effect}).
% --------------------------------------------------------------------------------------------------------------------------------
\subsection{Measurements and isometric operators. Quantum operations}
Postulate 3 and Postulate 4 describe two different ways of changing a system state.
It looks non-naturally.
Hence, we can set the problem: find a unified description for evolutions and measurements of a finite-level quantum system.

To solve the problem let's introduce for a state space $\mathcal{H}_n$ of an $n$-level quantum system and a finite set $X$
operators $J(x)\colon\mathcal{H}_n\rightarrow\mathcal{H}_n\otimes l^2(X)$ by the formula
\begin{equation}\label{eqn::jop}
J(x)\ket{\psi}=\ket{\psi}\otimes\ket{x},
\end{equation}
where $\ket{x}\in l^2(X)$ such that $\ket{x}(\cdot)=\delta(x,\cdot)$.

Properties of operators from the family $\{J(x)\colon x\in X\}$ are established by the next proposition, which is proved by the
direct calculation.
\begin{proposition}
Let $\mathcal{H}_n$ be a state space of an $n$-level quantum system and $X$ be a finite set, then the operators family
$\{J(x)\colon x\in X\}$ defined by formula {\rm(\ref{eqn::jop})} satisfies the next identities 
\begin{eqnarray}
&J(x)^\ast\sum\limits_{x'\in X}(\ket{\psi(x')}\otimes\ket{x'})=\ket{\psi(x)}\\
&J(x')^\ast J(x'')=\delta(x',x'')\cdot\matid\\
&J(x')J(x'')^\ast =\matid\otimes\ket{x'}\bra{x''}
\label{frm::prjdelta}
\end{eqnarray}
\end{proposition}

Now using a Kraus' family $\mathcal{K}=\{K(x)\colon x\in X\}$ for some measurement let's define an operator
$W_\mathcal{K}\colon\mathcal{H}_n\rightarrow\mathcal{H}_n\otimes l^2(X)$ by the formula
\begin{equation}\label{eqn::iso}
W_\mathcal{K}\ket{\psi}=\sum\limits_{x'\in X}(K(x')\ket{\psi}\otimes\ket{x'}).
\end{equation}
\begin{proposition}\label{prp::kraus2iso}
Let $\mathcal{K}=\{K(x)\colon x\in X\}$ be a Kraus' family, and $W_\mathcal{K}$ be the operator that is built by the formula {\rm(\ref{eqn::iso})}, then $W_\mathcal{K}$ is an isometric operator and the next identities hold for all $x\in X$:
\begin{equation}\label{eqn::krestore}
K(x)=J(x)^\ast W_\mathcal{K}
\end{equation}
\end{proposition}
\begin{proof}
Let $\ket{0},\dots,\ket{n-1}$ be some orthonormal basis in $\mathcal{H}_n$, then for any $k=0,\dots,n-1$ from (\ref{eqn::iso})
we have
\[
W_\mathcal{K}\ket{k}=\sum\limits_{x'\in X}(K(x')\ket{k}\otimes\ket{x'}).
\]
Hence,
\[
W_\mathcal{K}=\sum\limits_{k=0}^{n-1}\sum\limits_{x'\in X}(K(x')\ket{k}\otimes\ket{x'})\bra{k}
\]
and
\[
W_\mathcal{K}^\ast=\sum\limits_{k=0}^{n-1}\sum\limits_{x'\in X}\ket{k}(\bra{k}K(x')^\ast\otimes\bra{x'}).
\]
Therefore,
\begin{multline*}
W_\mathcal{K}^\ast W_\mathcal{K}=\\
\left(\sum\limits_{l=0}^{n-1}\sum\limits_{x''\in X}\ket{l}(\bra{l}K(x'')^\ast\otimes\bra{x''})\right)\cdot
\left(\sum\limits_{k=0}^{n-1}\sum\limits_{x'\in X}(K(x')\ket{k}\otimes\ket{x'})\bra{k}\right)=\\
\sum\limits_{k,l=0}^{n-1}\sum\limits_{x',x''\in X}\ket{l}\bra{l}K(x'')^\ast K(x')\ket{k}\braket{x''}{x'}\bra{k} =\\
\sum\limits_{k,l=0}^{n-1}\sum\limits_{x'\in X}\ket{l}\bra{l}K(x')^\ast K(x')\ket{k}\bra{k}.
\end{multline*}
Using the completeness condition one can obtain that
\[
W_\mathcal{K}^\ast W_\mathcal{K}=\sum\limits_{k,l=0}^{n-1}\ket{l}\braket{l}{k}\bra{k}=\sum\limits_{k=0}^{n-1}\ket{k}\bra{k}=
\matid.
\]
The last equation ensures that $W_\mathcal{K}$ is an isometric operator.\\
Equation (\ref{eqn::krestore}) is proved by the direct calculation:
\[
J(x)^\ast W_\mathcal{K}\ket{\psi}=J(x)^\ast\sum\limits_{x'\in X}(K(x')\ket{\psi}\otimes\ket{x'})=K(x)\ket{\psi}.
\]
Proof is complete\qed
\end{proof}

Using Proposition \ref{prp::kraus2iso} one can rewrite formulae (\ref{frm::prob}) and (\ref{frm::effect}) in the following way:
\begin{align}
\prob{x}{\rho}=\trace{\rho W_\mathcal{K}^\ast(\matid\otimes\ket{x}\bra{x})W_\mathcal{K}}
	\tag{\ref{frm::prob}$'$}\label{frm::prob_bis}\\
\effect{\rho}{x}=\frac{J(x)^\ast W_\mathcal{K}\rho W_\mathcal{K}^\ast J(x)}
	{\trace{\rho W_\mathcal{K}^\ast(\matid\otimes\ket{x}\bra{x})W_\mathcal{K}}}
		\tag{\ref{frm::effect}$'$}\label{frm::eff_bis}
\end{align}

Now we claim that this construction can be inverted.

Really, let $\mathcal{H}_n$ be a state space for an $n$-level quantum system, $X$ be a finite set of outcomes, and
$W\colon\mathcal{H}_n\rightarrow\mathcal{H}_n\otimes l^2(X)$ be an isometric operator.

Let's define a family $\mathcal{K}=\{K(x)\colon x\in X\}$ of operators on the space $\mathcal{H}_n$ by the formula
\begin{equation}\label{frm::kraus}
K(x)=J(x)^\ast W.
\end{equation}
\begin{proposition}\label{prp::iso2kraus}
Let $\mathcal{H}_n$ be a state space for an $n$-level quantum system, $X$ be a finite set of outcomes,
$W\colon\mathcal{H}_n\rightarrow\mathcal{H}_n\otimes l^2(X)$ be an isometric operator, and $\mathcal{K}=\{K(x)\colon x\in X\}$ be
the family of operators which is defined by formula {\rm(\ref{frm::kraus});} then
\begin{enumerate}
	\item
		$\mathcal{K}$ satisfies the completeness condition and, therefore,
		it is a Kraus' family; 
	\item
		$W_\mathcal{K}=W$. 
\end{enumerate}
\end{proposition}
\begin{proof}
To prove the completeness condition let's calculate the left side of (\ref{frm::comp}) using (\ref{frm::prjdelta}) and the isometry property 
\begin{multline*}
\sum\limits_{x\in X}K(x)^\ast K(x)=\\
\sum\limits_{x\in X}W^\ast J(x)J(x)^\ast W=\sum\limits_{x\in X}W^\ast(\matid\otimes\ket{x}\bra{x})W=W^\ast W=\matid.
\end{multline*}
To prove the second statement let's calculate using (\ref{frm::prjdelta})
\begin{multline*}
W_\mathcal{K}\ket{\psi} =\\
\sum\limits_{x\in X}(K(x)\ket{\psi}\otimes\ket{x})=\sum\limits_{x\in X}J(x)K(x)\ket{\psi}=
\sum\limits_{x\in X}J(x)(J^\ast(x)W)\ket{\psi}=\\
\sum\limits_{x\in X}(J(x)J^\ast(x))W\ket{\psi}=\sum\limits_{x\in X}(\matid\otimes\ket{x}\bra{x})W\ket{\psi}=W\ket{\psi}.
\end{multline*}
Proof is complete\qed
\end{proof}

Proposition \ref{prp::kraus2iso} and \ref{prp::iso2kraus}, formulae (\ref{frm::prob_bis}) and (\ref{frm::eff_bis}) substantiate
replacing the Kraus' families by the corresponding isometric operators under studying the interaction of
quantum systems with classical systems.
This replacing leads us to unification of Postulate 3 and Postulate 4.
To stress such unification we will say that an isomeric operator describes the quantum operation by formulae (\ref{frm::prob_bis}) and (\ref{frm::eff_bis}).
\begin{definition}
Let $\mathcal{H}_n$ be a state space of an $n$-level quantum system, $X$ be a finite set of outcomes, then isometric operators $W_1,W_2\colon\mathcal{H}_n\rightarrow\mathcal{H}_n\otimes l^2(X)$ are called equivalent if for all $x\in X$ and for any density
operator $\rho$ the following equalities are true
\begin{align}
	\trace{\rho W_1^\ast(\matid\otimes\ket{x}\bra{x})W_1}&=\trace{\rho W_2^\ast(\matid\otimes\ket{x}\bra{x})W_2},
		\label{frm::probinv}\\
	J(x)^\ast W_1\rho W_1^\ast J(x)&=J(x)^\ast W_2\rho W_2^\ast J(x).\label{frm::effinv}
\end{align}

Classes of this equivalence will be called {\bfseries quantum operations} with a set of outcomes $X$. 
\end{definition}

Easy to see that isometric operators $W_1,W_2\colon\mathcal{H}_n\rightarrow\mathcal{H}_n\otimes l^2(X)$ describe the same quantum
operation if $J(x)^\ast W_2=\mathrm{e}^{i\theta(x)}J(x)^\ast W_1$ for any $\theta\colon X\rightarrow [0,2\pi)$.

We claim that the inverse statement is true too.
\begin{theorem}\label{thr::eqv}
Let $\mathcal{H}_n$ be a state space of an $n$-level quantum system, $X$ be a finite set of outcomes, $W_1,W_2\colon\mathcal{H}_n\rightarrow\mathcal{H}_n\otimes l^2(X)$ be equivalent isometric operators then
$J(x)^\ast W_2=\mathrm{e}^{i\theta(x)}J(x)^\ast W_1$ for some $\theta\colon X\rightarrow [0,2\pi)$. 
\end{theorem}
\begin{proof}
It is evident, that each isometric operator $W_s\colon\mathcal{H}_n\rightarrow\mathcal{H}_n\otimes l^2(X)$, where $s=1,2$, can be represented by the formula
\begin{equation}\label{frm::isovrep}
W_s=\sum\limits_{x\in X}\sum\limits_{k=0}^{n-1}(\ket{\omega_k^{(s)}(x)}\otimes\ket{x})\bra{k},
\end{equation}
where $\{\ket{0},\dots,\ket{n-1}\}$ is an orthonormal basis in $\mathcal{H}_n$ and
$\ket{\omega_k^{(s)}(x)}=J(x)^\ast W_s\ket{k}$ for $k=0,\dots,n-1$ and $x\in X$.

Using representation (\ref{frm::isovrep}) we calculate $\mathrm{Pr}(x\mid \ket{k}\bra{k})$ for $W_s$ where $s=1,2$.
\begin{multline*}
\mathrm{Pr}(x\mid \ket{k}\bra{k})=\\
\trace{\ket{k}\bra{k}W_s^\ast(\matid\otimes\ket{x}\bra{x})W_s}=\bra{k}W_s^\ast(\matid\otimes\ket{x}\bra{x})W_s)\ket{k}=\\
\bra{k}W_s^\ast(\matid\otimes\ket{x}\bra{x})\sum\limits_{x'\in X}\sum\limits_{l=0}^{n-1}
(\ket{\omega_l^{(s)}(x')}\otimes\ket{x'})\braket{l}{k}=
\bra{k}W_s^\ast(\ket{\omega_k^{(s)}(x)}\otimes\ket{x})=\\
\bra{k}\sum\limits_{x'\in X}\sum\limits_{l=0}^{n-1}\ket{l}(\bra{\omega_l^{(s)}(x')}\otimes\bra{x'})
(\ket{\omega_k^{(s)}(x)}\otimes\ket{x})=
\|\omega_k^{(s)}\|^2.
\end{multline*}
Using this and identity (\ref{frm::probinv}) one can derive that for all $k=0,\dots,n-1$ and $x\in X$ the next equality holds
\begin{equation}\label{frm::normeqv}
\|\omega_k^{(1)}(x)\|^2=\|\omega_k^{(2)}(x)\|^2.
\end{equation}

Let $I_s(x)=\{k\mid 0\leq k<n, \|\omega_k^{(s)}(x)\|\neq 0\}$ for each $x\in X$ and $s=0,1$.
From (\ref{frm::normeqv}) it follows that $I_1(x)=I_2(x)$, hence, we can denote this set by $I(x)$.

From (\ref{frm::effinv}) one can derive that for all $x\in X$ and $k\in I(x)$
\begin{equation}\label{eqv::temp}
\ket{\omega_k^{(2)}(x)}\bra{\omega_k^{(2)}(x)}=\ket{\omega_k^{(1)}(x)}\bra{\omega_k^{(1)}(x)}.
\end{equation}
The next equality is obtained by multiplying equality (\ref{eqv::temp}) from left by $\bra{\omega_k^{(1)}(x)}$ and from right by
$\ket{\omega_k^{(1)}(x)}$ and using equality (\ref{frm::normeqv}):
\begin{equation}\label{eqv::temp1}
|\braket{\omega_k^{(2)}(x)}{\omega_k^{(1)}(x)}|^2=\|\omega_k^{(2)}(x)\|^2\cdot\|\omega_k^{(1)}(x)\|^2.
\end{equation}
From the (\ref{eqv::temp1}) and (\ref{frm::normeqv}) it follows that for all $x\in X$ and $k\in I(x)$
\begin{equation}
\ket{\omega_k^{(2)}(x)}=\mathrm{e}^{i\theta(k,x)}\ket{\omega_k^{(1)}(x)},\mbox{ where }0\leq\theta(k,x)<2\pi.
\end{equation}

Further, from (\ref{frm::effinv}) it follows that for all $x\in X$ and $k,l\in I(x)$ the next equality is true:
\[
\ket{\omega_k^{(2)}(x)}\bra{\omega_l^{(2)}(x)}=\ket{\omega_k^{(1)}(x)}\bra{\omega_l^{(1)}(x)}.
\]
Therefore,
\[
\mathrm{e}^{i(\theta(k,x)-\theta(l,x))}\ket{\omega_k^{(1)}(x)}\bra{\omega_l^{(1)}(x)}=
\ket{\omega_k^{(1)}(x)}\bra{\omega_l^{(1)}(x)},
\]
and $\mathrm{e}^{i(\theta(k,x)-\theta(l,x))}=1$.

In summary, we obtain the next equality for all $x\in X$ and $k\in I(x)$
\begin{equation}\label{eqv::temp2}
\ket{\omega_k^{(2)}(x)}=\mathrm{e}^{i\theta(x)}\ket{\omega_k^{(1)}(x)}.
\end{equation}
Using (\ref{eqv::temp2}) for $x\in X$, $k\in I(x)$ and the equality $\ket{\omega_l^{(2)}(x)}=\ket{\omega_l^{(1)}(x)}=0$
for $l\in \{0,\dots, n-1\}\setminus I(x)$ one can get that equality (\ref{eqv::temp2}) is true for all $0\leq k<n$.

Therefore, $J(x)^\ast W_2=\mathrm{e}^{i\theta(x)}J(x)^\ast W_1$ for some $\theta\colon X\rightarrow [0,2\pi)$\qed
\end{proof}
\begin{corollary}
Two isometric operators $W_1,W_2\colon\mathcal{H}_n\rightarrow\mathcal{H}_n\otimes l^2(X)$ define the same quantum operation
if and only if for some $\theta\colon X\rightarrow [0,2\pi)$ the following equality holds
\[
W_2=\Theta W_1\mbox{, where }\Theta=\matid\otimes\sum_{x\in X}e^{i\theta(x)}\ket{x}\bra{x}.
\]
\end{corollary}
% --------------------------------------------------------------------------------------------------------------------------------
\section{Abstract Quantum Automata}
Now we describe some class of mathematical models for quantum information processes.
This class we call the class of abstract quantum automata.
% --------------------------------------------------------------------------------------------------------------------------------
\subsection{The notion of an abstract quantum automaton}
\begin{definition}
Let $\mathcal{H}_m (m>1)$ be a state space of an $m$-level quantum system, and let 
$\mathcal{G}=(\mathcal{N},\mathcal{A},\dom,\codom,n_0,F)$ be a control graph.
Suppose that each non-terminal node $n$ of the graph $\mathcal{G}$ is connected with a quantum operation for which $\mathcal{H}_m$
is the state space, $\mathrm{Out}(n)$ is the outcomes set, and
$W_n\colon\mathcal{H}_m\rightarrow\mathcal{H}_m\otimes l^2(\mathrm{Out}(n))$ is an isometric operator describing the operation.
Then the tuple
\[
(\mathcal{H}_m,\mathcal{G},\{W_n\mid n\in\mathcal{N}\setminus F\})
\]
is called an {\bfseries abstract quantum automaton}.
\end{definition}

The next definition describes the set of runs for an abstract quantum automaton similarly to Definition \ref{dfn::rrun}.
\begin{definition}
Let $(\mathcal{H}_m,\mathcal{G},\{W_n\mid n\in\mathcal{N}\setminus F\})$ be an abstract quantum automaton where the control graph
$\mathcal{G}$ is equal to $(\mathcal{N},\mathcal{A},\dom,\codom,n_0,F)$.
Then a partial map $C\colon\mathbb{N}\dashrightarrow\mathcal{N}\times\mathfrak{S}$ is called a {\bfseries run} of the automaton if
it satisfies the following conditions
\begin{itemize}
	\item[{\rm(i)}]
		$C(0)=(n_0,\rho)\mbox{ where }\rho\in\mathfrak{S}_m${\rm;}
	\item[{\rm(ii)}]
		if $C(t)\neq\emptyset$ for some $t\in \mathbb{N}$ then $C(t')\neq\emptyset$ for all $t'\in \mathbb{N}$ such that $t'<t$
		{\rm;}
	\item[{\rm(iii)}]
		if $\emptyset\neq C(t+1)=(n',\rho')$ for some $t\in \mathbb{N}$ and $C(t)=(n,\rho)$ then there exists
		$a\in\mathrm{Out}(n)$ such that $\prob{a}{n,\rho}=\trace{\rho W_n^\ast(\matid\otimes\ket{a}\bra{a})W_n}>0$,
		$n'=\codom(a)$, and $\rho'=\effect{\rho}{n,a}=\dfrac{J(a)^\ast W_n\rho W_n^\ast J(a)}{\prob{a}{n,\rho}}${\rm;} 
	\item[{\rm(iv)}]
		if $\emptyset\neq C(t)=(n,\rho)$ for $n\in F$ and $\rho\in\mathfrak{S}_m$ then $C(t+1)=\emptyset$.
\end{itemize}
\end{definition}
\begin{example}
Let's consider a quantum information process which sets a qubit (2-level quantum system) into the state $\ket{0}\bra{0}$.

Evidently, that this problem can not be solved by any unitary transformation.

We shall specify an abstract quantum automaton that does it.
The control graph of the automaton is shown in Fig. \ref{fig::qucleaner}.
\begin{figure}[h]
\centering
\includegraphics[width=0.7\textwidth]{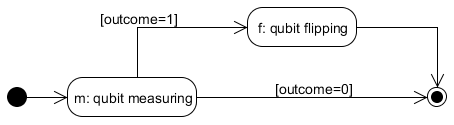}
\caption{Qubit cleaner}
\label{fig::qucleaner}
\end{figure}

As one can see $\mathrm{Out}(m)=\{0,1\}$.
Let's define $W_m\colon\mathcal{H}_2\rightarrow\mathcal{H}_2\otimes l^2(\{0,1\})$ by the formula
\[
W_m\ket{\psi}=\ket{0}\braket{0}{\psi}\otimes\ket{0}+\ket{1}\braket{1}{\psi}\otimes\ket{1}.
\]
Further, $\mathrm{Out}(f)$ is a singleton hence $W_f\colon\mathcal{H}_2\rightarrow\mathcal{H}_2$.
Let's define
\[
W_f\ket{\psi}=\ket{0}\braket{1}{\psi}+\ket{1}\braket{0}{\psi}.
\]
Easy to see that for an arbitrary initial state of a qubit its state after handling by the automaton is equal to $\ket{0}\bra{0}$.

Therefore, we have built the abstract quantum automaton that specifies the process of cleaning a qubit.
\end{example}
\begin{example}
This example deals with preparing an entangled pair of qubits.

We shall specify an abstract quantum automaton that does it.
The control graph of the automaton is shown in Fig. \ref{fig::entanglement}.
\begin{figure}[h]
\centering
\includegraphics[width=0.8\textwidth]{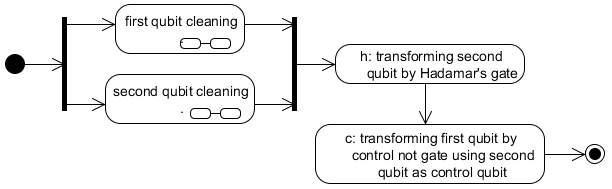}
\caption{Preparing an entangled pair of qubits}
\label{fig::entanglement}
\end{figure}

Let's define $W_h\colon\mathcal{H}_2\otimes\mathcal{H}_2\rightarrow\mathcal{H}_2\otimes\mathcal{H}_2$ by the formulae
\begin{eqnarray}
W_h(\ket{\psi}\otimes\ket{0})=\ket{\psi}\otimes\dfrac{1}{\sqrt{2}}(\ket{0}+\ket{1}),\notag\\
W_h(\ket{\psi}\otimes\ket{1})=\ket{\psi}\otimes\dfrac{1}{\sqrt{2}}(\ket{0}-\ket{1}).\notag
\end{eqnarray}
and define $W_c\colon\mathcal{H}_2\otimes\mathcal{H}_2\rightarrow\mathcal{H}_2\otimes\mathcal{H}_2$ by the formulae
\begin{eqnarray}
W_c(\ket{\psi}\otimes\ket{0})=&\ket{\psi}\otimes\ket{0},\notag\\
W_c(\ket{\psi}\otimes\ket{1})=&\ket{1}\braket{\psi}{0}\otimes\ket{1}+\ket{0}\braket{\psi}{1}\otimes\ket{1}.\notag
\end{eqnarray}
Easy to see that for an arbitrary initial state of a qubit pair its state after handling by the automaton is equal to
$\dfrac12(\ket{0}\otimes\ket{0}+\ket{1}\otimes\ket{1})(\bra{0}\otimes\bra{0}+\bra{1}\otimes\bra{1})$.
\end{example}

These examples demonstrate that modelling of quantum information processes by abstract quantum automata allows to describe
processes of initial preparation of a quantum memory for quantum computing devices.
% --------------------------------------------------------------------------------------------------------------------------------
\subsection{Quantum teleportation as an abstract quantum automaton}
To complete the paper let's consider the quantum teleportation process and let's show that it can be described by an abstract
quantum automaton.

Quantum teleportation is a process by which a qubit state can be transmitted exactly from one location to another, without the
qubit being transmitted through the intervening space.
This phenomenon has been confirmed experimentally \cite{bib::bou,bib::bos}.

The control graph of the automaton is shown in Fig. \ref{fig::teleportation}.
\begin{figure}[h]
\centering
\includegraphics[width=0.85\textwidth]{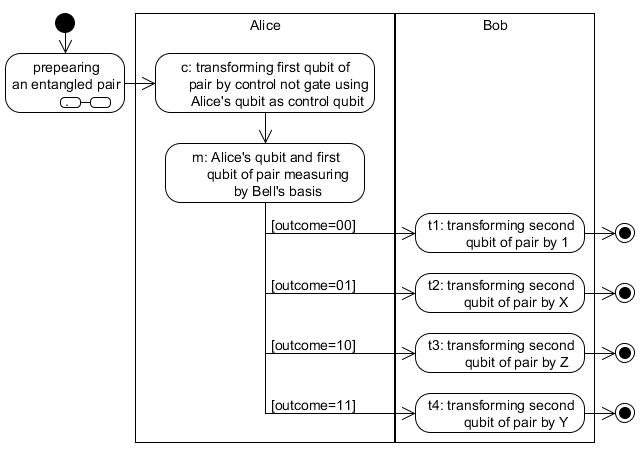}
\caption{Teleportation}
\label{fig::teleportation}
\end{figure}

Let's define
$W_c\colon\mathcal{H}_2\otimes\mathcal{H}_2\otimes\mathcal{H}_2\rightarrow
\mathcal{H}_2\otimes\mathcal{H}_2\otimes\mathcal{H}_2$, where the first qubit is Alice's qubit, the second and the third qubits
are the first and the second qubits of the entangled pair respectively, by the formulae
\begin{eqnarray}
W_c(\ket{0}\otimes\ket{0}\otimes\ket{\psi})=&\ket{0}\otimes\ket{0}\otimes\ket{\psi},\notag\\
W_c(\ket{0}\otimes\ket{1}\otimes\ket{\psi})=&\ket{0}\otimes\ket{1}\otimes\ket{\psi},\notag\\
W_c(\ket{1}\otimes\ket{0}\otimes\ket{\psi})=&\ket{1}\otimes\ket{1}\otimes\ket{\psi},\notag\\
W_c(\ket{1}\otimes\ket{1}\otimes\ket{\psi})=&\ket{1}\otimes\ket{0}\otimes\ket{\psi}.\notag
\end{eqnarray}
Further, $\mathrm{Out}(m)=\{00,01,10,11\}$ and the corresponding isometric operator is defined by formulae
\begin{eqnarray}
W_m(\ket{0}\otimes\ket{0}\otimes\ket{\psi})=&\ket{0}\otimes\ket{0}\otimes\ket{\psi}\otimes\ket{00},\notag\\
W_m(\ket{0}\otimes\ket{1}\otimes\ket{\psi})=&\ket{0}\otimes\ket{1}\otimes\ket{\psi}\otimes\ket{01},\notag\\
W_m(\ket{1}\otimes\ket{0}\otimes\ket{\psi})=&\ket{1}\otimes\ket{0}\otimes\ket{\psi}\otimes\ket{10},\notag\\
W_m(\ket{1}\otimes\ket{1}\otimes\ket{\psi})=&\ket{1}\otimes\ket{1}\otimes\ket{\psi}\otimes\ket{11}.\notag
\end{eqnarray}
By direct calculation one can prove that the initial state $\ket{\psi}\bra{\psi}\otimes\rho$ for an arbitrary
$\rho\in\mathfrak{S}_4$ is transformed by the automaton into the state
$\dfrac14(\ket{0}\bra{0}\otimes\ket{0}\bra{0}+\ket{0}\bra{0}\otimes\ket{1}\bra{1}+\ket{1}\bra{1}\otimes\ket{0}\bra{0}+
\ket{1}\bra{1}\otimes\ket{1}\bra{1})\otimes\ket{\psi}\bra{\psi}$.
% --------------------------------------------------------------------------------------------------------------------------------
\section*{Conclusion}
Summarising the above we can conclude:
\begin{itemize}
	\item[$-$]
		our attempt to solve the Yu. Manin's problem led us towards the notion of an abstract quantum automaton;
	\item[$-$]
		this notion is based on a computational model known as a machine of A. Kolmogorov and V. Uspensky;
	\item[$-$]
		abstract quantum automata can be used for formal specification of quantum information processes including non-invertible
		processes like qubit cleaning, entangled pair preparing and quantum teleportation.
\end{itemize}

The authors know that quantum algorithms can be specified by using abstract quantum automata but corresponding results are not
given in the paper because they are cumbersome.
% --------------------------------------------------------------------------------------------------------------------------------

\end{document}